\newif\ifelsarticle
\ifelsarticle
\documentclass{elsarticle}
\else
\documentclass{scrartcl}
\fi

\usepackage{amsmath,amsthm,amssymb,stmaryrd,mathrsfs}
\usepackage{color}
\usepackage{ifthen}

\usepackage[all]{xy}

\newtheorem{theorem}{Theorem}
\newtheorem{lemma}[theorem]{Lemma}
\newtheorem{proposition}[theorem]{Proposition}
\newtheorem{corollary}[theorem]{Corollary}
\theoremstyle{definition}
\newtheorem{definition}[theorem]{Definition}
\newtheorem{example}[theorem]{Example}
\newtheorem{remark}[theorem]{Remark}

\renewcommand{\P}{\mathcal{P}}
\newcommand{\Pc}{\overline{\mathcal{P}}}
\newcommand{\PP}{\mathcal{P}^2}
\newcommand{\PPc}{\overline{\mathcal{P}}^2}

\newcommand{\D}{\mathcal{D}}
\newcommand{\Sets}{\mathbf{Sets}}

\newcommand{\SDD}{\mathcal{S}}
\newcommand{\ZSDD}{\mathcal{Z}}

\newcommand{\BDDI}{\beta}
\newcommand{\ZDDI}{\zeta}
\newcommand{\SDDI}{\sigma}
\newcommand{\ZSDDI}{\xi}

\newcommand{\emptycomb}{\varepsilon}

\newcommand{\zeroterm}{\mathbf{0}}
\newcommand{\oneterm}{\mathbf{1}}
\newcommand{\true}{\mathit{true}}
\newcommand{\false}{\mathit{false}}

\newcommand{\CatC}{\mathscr{C}}

\newcommand{\TO}{\mathbf{TO}}
\newcommand{\VTree}{\mathbf{VTree}}

\title{BDDs Naturally Represent Boolean Functions, and ZDDs Naturally Represent Sets of Sets}

\ifelsarticle
\else
\author{Kensuke Kojima}
\fi

\begin{document}

\ifelsarticle
\input{front_elsarticle}
\else
\maketitle
\begin{abstract}
  This paper studies a difference between Binary Decision Diagrams (BDDs) and Zero-suppressed BDDs (ZDDs) from a conceptual point of view.
  It is commonly understood that a BDD is a representation of a Boolean function, whereas a ZDD is a representation of a set of sets.
  However, there is a one-to-one correspondence between Boolean functions and sets of sets, and therefore we could also regard a BDD as a representation of a set of sets, and similarly for a ZDD and a Boolean function.
  The aim of this paper is to give an explanation why the distinction between BDDs and ZDDs mentioned above is made despite the existence of the one-to-one correspondence.  To achieve this, we first observe that Boolean functions and sets of sets are equipped with non-isomorphic functor structures, and show that these functor structures are reflected in the definitions of BDDs and ZDDs.  This result can be stated formally as naturality of certain maps.  To the author's knowledge, this is the first formally stated theorem that justifies the commonly accepted distinction between BDDs and ZDDs.  In addition, we show that this result extends to sentential decision diagrams and their zero-suppressed variant.
  % \todo{contributions: (1) informal にしかしられていなかったことをちゃんと定式化したこと (2) 同じ主張が OBDD, SDD などにも（適当な定式化のもとで）成立することを示したこと}
\end{abstract}

\fi

\section{Introduction}
\label{sec:intro}

% Boolean functions and combination sets can be seen as two ways of seeing a collection of combinations.  \todo{write the explanation: used in combinatorial problems such as enumeration problem (applied to various practical problems), and are in one-to-one correspondence (equivalent in this sense)}

% Collections of objects are central objects of interest in combinatorial problems.  We study a conceptual aspect of several representations of such data.  Therefore we focus on the representations of combination sets, which is a set of subsets of a fixed (finite) set.  It is well-known that combination sets are in one-to-one correspondence between Boolean functions.
In this paper, we study a difference between two types of decision diagrams: Binary Decision Diagrams (BDDs, for short)~\cite{Lee59,Akers78}, and Zero-suppressed BDDs (ZDDs, for short)~\cite{Minato:1993:ZBS:157485.164890}.
It is commonly understood that a BDD is a representation of a Boolean function (a function that takes several, fixed number of Boolean values and returns a Boolean value), and a ZDD is that of a combination set (a family of subsets of a fixed set).  This fact implies that Boolean functions and combination sets are considered different.  However, it is easy to find a one-to-one correspondence between them, and by using this correspondence, one may regard a representation of a Boolean function as that of a combination set, and vice versa.  This would mean that, contrary to the claim above, both BDDs and ZDDs can be used to represent both Boolean functions and combination sets.  This argument leads to the following question: what distinguishes BDDs (or Boolean functions) and ZDDs (or combination sets)?  Why are they considered different, despite the existence of the one-to-one correspondence?  The aim of this paper is to answer this question from a conceptual point of view.

% \todo{describe a brief history~\cite{Minato13:TechniquesBddZdd}; ZDD is proposed as a better representation of a set than BDD}
% Using this correspondence, a Boolean function taking $n$ inputs can be regarded as a set of combinations from $n$ elements.

% Indeed, a BDD, which is originally developed as a representation of a Boolean function, can be used to represent a combination set.  However, there has been an inconvenience in doing so, and that was the primary motivation for inventing a ZDD~\cite{Minato:1993:ZBS:157485.164890,Minato13:TechniquesBddZdd}.
% Therefore, it is commonly understood that a BDD is a representation of a Boolean function, whereas a ZDD represent combination sets.
% Here, a natural question arises: what distinguishes BDDs (or Boolean functions) from ZDDs (or combination sets)?  Why are they considered different, despite the existence of a one-to-one correspondence?
% \todo{actually, we have classified into two class according to this distinction, so the actual order of explanation is the opposite}

An informal explanation for this has already been given when ZDDs are proposed~\cite{Minato:1993:ZBS:157485.164890}.
% a problem with representing combination sets using a BDD.
% Technically, the difference lies in the treatment of elements that do not appear in a diagram (called \emph{default variables} in~\cite{Minato:1993:ZBS:157485.164890}).
If a certain variable does not appear in a representation of a Boolean function, this variable is not used by the function, and therefore its value should be irrelevant to the output of the function.
% However, when considering a combination set, it is argued that such an element should be treated as absent from any combination.\todo{たぶんちょっとちがう}  The same combination set has different representations in BDDs depending on the set of elements (variables) being considered.  \todo{probably this makes sense only if the translation between Boolean functions and combination sets are understood}
This means that, if a combination set is represented by a BDD by identifying it with the corresponding Boolean function, then the representation changes when we extend the set of elements (variables) being considered.  This is because extra elements have to be explicitly excluded.  For example, if $a$ and $b$ are the only elements being considered, a combination set $\{\{a\}, \{b\}\}$ can be expressed by the Boolean function $a \oplus b$ (where $\oplus$ is XOR), but if another element $c$ is present in the context, the same combination set should be represented by $(a \oplus b) \land \neg c$; an extra element $c$ has to be excluded explicitly.  Minato~\cite{Minato:1993:ZBS:157485.164890} pointed out that this dependence on extra elements as an inconvenience in representing a combination set using a BDD, and addressed this problem by introducing a new representation, a ZDD, that is not affected by extra elements.

% \todo{I might want to mention ordering, vtrees and partition, for a justification for considering subfunctors.  (1) Ordering on variables are initially omitted, but actually it can be treated by slightly extending the framework (by generalizing the domain of diagram/predicate functors to a general category).  Actually, in our cases, it cat be treated as a subfunctor of the ``free'' version up to a suitable forgetful functor, and the naturality result easily follows from the general observation.  (2) Vtrees can be treated similarly, and in addition partition condition and its variant also fit into this extended framework.  This consideration brings us a new insight into the naturality of the notion of implicit partitioning introduced by Nishino et al.~\cite{Nishino:2016:ZSD:3015812.3015969}.}
% \todo{Also, it should be noted that ordered ZDD is the ordinary ZDD, and unordered (free) ZDD is not usually considered.}

The goal of this paper is to reformulate this explanation as a more formally stated theorem.
Our formulation uses the language of category theory.  We first observe that Boolean functions and combination sets are equipped with essentially different functor structures: although they are in one-to-one correspondence, the bijection cannot be a natural isomorphism.  We next show that the difference between those functor structures
are reflected in the definition of BDDs and ZDDs.  This can be formally stated as the naturality of the semantics of BDDs and ZDDs (here, a semantics is given by a function that receives a decision diagram and returns the mathematical object it represents).
In this sense, a BDD is a natural representation of a Boolean function (but not a combination set), whereas a ZDD is that of a combination set (but not a Boolean function).
In addition, we also consider \emph{Sentential Decision Diagrams (SDDs)}~\cite{Darwiche:2011:SNC:2283516.2283536} and \emph{Zero-suppressed SDDs (ZSDDs)}~\cite{Nishino:2016:ZSD:3015812.3015969}, which extend BDDs and ZDDs, respectively, and show that analogous result holds for these data structures as well.
We believe that our results provide a better understanding of these data structures by uncovering a mathematical structure behind them.

The rest of this paper is organized as follows.
Section~\ref{sec:csbf} introduces several basic notions and notations.
Section~\ref{sec:bdd} introduces BDDs and ZDDs, their semantics, and shows that the semantics are natural with respect to appropriate functor structures.  Both unordered and ordered cases are discussed.
Section~\ref{sec:sdd} discusses SDDs and ZSDDs.  After introducing their definitions, a naturality result similar to the previous section is proved.  The notions of vtrees and partition are also discussed.
Finally, Section~\ref{sec:concl} concludes the paper.

% This paper is an extended and revised version of the author's previous
% publication presented in  workshop~\cite{FPAI-B509-en}.

\section{Combination Sets and Boolean Functions}
\label{sec:csbf}

\begin{definition}
  \label{def:comb}
  Let $X$ be a set.  A \emph{combination} over $X$ is just a subset of $X$, and a \emph{combination set} over $X$ is a set of combinations over $X$.  In other words, a combination and a combination set are elements of $\P(X)$ and $\PP(X)$, respectively.
\end{definition}

\begin{example}
  Let $G$ be a graph, and $E$ the set of all edges in $G$.
  % Then, a subgraph of $G$ is determined by a subset of $E$, i.e., a combination over $E$.  Therefore, the set of subgraphs satisfying a certain condition (e.g., paths between fixed nodes $s$ and $t$, spanning trees, and Hamiltonian circuits) are examples of combination sets over $E$.
  Then, the set of all paths between two fixed nodes $s$ and $t$ is an example of a combination set over $E$, because a path is represented by a subset of $E$.  Similarly, the set of all spanning trees of $G$ and the set of all Hamiltonian circuits of $G$ are combination sets over $E$.
\end{example}

If no confusion arises, we denote a combination by a sequence.  For example, we write $abc$ for $\{a, b, c\}$, and $\{a, b, ab\}$ for $\{\{a\}, \{b\}, \{a, b\}\}$.  We also use $\emptycomb$ for the empty combination, although formally it is the same as $\emptyset$, to emphasize that we regard it as a combination.

\begin{definition}
  Let $X$ be a set.  A \emph{Boolean function} over $X$ is an element of $2^{2^X}$, where $2 = \{0, 1\}$.
\end{definition}

It is well-known that, for finite $X$, a Boolean function over $X$ can always be represented by a Boolean formula whose atoms are taken from $X$.  For example, if $a, b \in X$, then $a \lor b$ denotes the function which takes $f \in 2^X$ and returns $1$ if and only if either $f(a) = 1$ or $f(b) = 1$.  Below we sometimes use formulas to represent Boolean functions.

There is a one-to-one correspondence between combination sets and Boolean functions.

\begin{proposition}
  \label{prop:cs-to-bf}
  A map $\tau_X\colon \PP(X) \to 2^{2^X}$ defined below is a bijection:
  \begin{equation*}
    \tau_X(P)(f) =
    \begin{cases}
      1 & f^{-1}(1) \in P, \\
      0 & f^{-1}(1) \notin P.
    \end{cases}
  \end{equation*}
\end{proposition}

Despite the existence of a bijective correspondence, we distinguish combination sets and Boolean functions as essentially different objects, by introducing different functor structures.

Below we denote by $\P$ and $\Pc$ the \emph{covariant} and \emph{contravariant} power set functors, respectively (that is, for $f\colon X \to Y$, their morphism parts are given by $\P(f)(A) = f(A)$ for $A \in \P(X)$, and $\Pc(f)(B) = f^{-1}(B)$ for $B \in \P(Y)$).

We call $\PP = \P \circ \P$ the \emph{combination sets functor}.  It is easy to see that there is a natural isomorphism $\Pc(X) \simeq 2^X$, and $\tau_X$ above is in fact induced from this isomorphism: $\tau_X\colon \PPc(X) \to 2^{2^X}$.  For this reason, we call $\PPc$ the \emph{Boolean functions functor}, and identify $\PPc(X)$ and $2^{2^X}$ from now on.

It is easy to check that $\tau_X$ above is not natural.  Moreover, we can prove that there is no natural isomorphism between $\PP$ and $\PPc$.  In this sense, combination sets and Boolean functions are equipped with different structures, and the difference originates from the two distinct functor structures on power sets.

\begin{proposition}
  $\PP$ and $\PPc$ are not isomorphic as functors.
\end{proposition}
\begin{proof}
  % \todo{do we want to include this proof in the paper?}
  Let $X = \{x, y\}$ (where $x \ne y$), and consider the inclusion $i \colon \emptyset \to X$ and $r \colon X \to X$ defined by $r(x) = r(y) = x$.  We prove that if $\alpha\colon \PPc \to \PP$ is natural, then $\alpha_X(\P(X)) = \alpha_X(\{\emptycomb, x, xy\})$, and therefore $\alpha_X$ is not injective.  Consider the following commutative diagram.
  \begin{equation*}
    \xymatrix{
      \PPc(\emptyset)\ar[r]^{\PPc(i)} \ar[d]_{\alpha_\emptyset}
      & \PPc(X) \ar[d]_{\alpha_X}
      & \PPc(X) \ar[l]_{\PPc(r)} \ar[d]_{\alpha_X} \\
      \PP(\emptyset) \ar[r]_{\PP(i)}
      & \PP(X)
      & \PP(X) \ar[l]^{\PP(r)}
    }
  \end{equation*}
  Let $P = \alpha_\emptyset(\{\emptycomb\})$ and $Q = \alpha_X(\{\emptycomb, x, xy\})$.  By the naturality of $\alpha$ and the definition of $\PP(i)$, we have
  \begin{equation*}
    P = \PP(i)(P) = \PP(i)(\alpha_\emptyset(\{\emptycomb\})) = \alpha_X(\PPc(i)(\{\emptycomb\})) = \alpha_X(\P(X)).
  \end{equation*}
  Similarly we have
  \begin{equation*}
    \PP(r)(Q) = \alpha_X(\PPc(r)(\{\emptycomb, x, xy\})) = \alpha_X(\P(X)),
  \end{equation*}
  and thus $\PP(r)(Q) = P$.  Therefore, to prove $P = Q$, it suffices to show that $\PP(r)(Q) = Q$.  From $P = \alpha_\emptyset(\{\emptycomb\}) \in \PP(\emptyset)$ we obtain $P \subseteq \{\emptycomb\}$, and thus $\PP(r)(Q) \subseteq \{\emptycomb\}$.  This is possible only if $Q \subseteq \{\emptycomb\}$, and in such a case, it is clear from the definition of $\PP(r)$ that $\PP(r)(Q) = Q$.
\end{proof}

\section{Binary Decision Diagrams}
\label{sec:bdd}

\subsection{Definition of BDDs and ZDDs}

A BDD and a ZDD are both graph representations of a combination set (or a Boolean function, via the bijection in Proposition~\ref{prop:cs-to-bf}).
We first introduce a class of directed graphs, which we call diagrams, and can be regarded as both BDDs and ZDDs.  We define two interpretations of diagrams, one as BDDs and the other as ZDDs.
% It is common that an order among variables are introduced when considering BDD or ZDD representation; this is discussed later in Section~\ref{subsec:odd}.

\begin{definition}
  A \emph{diagram} over a set $X$ is a rooted, directed acyclic graph
  with labeled nodes and edges satisfying the following.
  \begin{itemize}
  \item There are two types of nodes: \emph{decision nodes} and \emph{terminal
    nodes}.
  \item Each decision node is labeled by an element of $X$, and has two outgoing edges.  One of the edges is labeled by $0$ and another is labeled by $1$, and called a \emph{$0$-edge} and a \emph{$1$-edge}, respectively.
  \item Each terminal node is labeled by either $0$ or $1$, and has no outgoing edges.  A terminal node is called either a \emph{$0$-terminal node} or a \emph{$1$-terminal node} according to its label.
  \end{itemize}
  We write $\D(X)$ for the set of all diagrams over a set $X$.
\end{definition}

% \begin{wrapfigure}{r}{4cm}
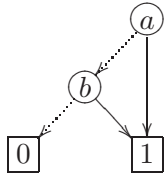
\begin{figure}[t]
  \centering
  \begin{equation*}
  % \(
    \xymatrix@=1pc{
      & & *+[o][F]{\vphantom{b} a} \ar@{.>}[dl] \ar[dd] \\
      & *+[o][F]{b} \ar@{.>}[dl] \ar[dr] \\
      *+[F]{0} & & *+[F]{1}}
  % \)
  \end{equation*}
  \caption{An example of a diagram.  Circles, squares, dotted arrows, and solid arrows denote decision nodes, terminal nodes, $0$-edges, and $1$-edges, respectively.}
  \label{fig:ex_dd}
\end{figure}
% \end{wrapfigure}

A diagram is often depicted as in Figure~\ref{fig:ex_dd}.
For convenience, we write $\zeroterm$ and $\oneterm$ for $0$- and $1$-terminal nodes, respectively, and $(a, F,
G)$ for a decision node which is labeled by $a$ and has $0$-edge and $1$-edge pointing to $F$ and $G$, respectively.  For example, the diagram in Figure~\ref{fig:ex_dd} is written as $(a, (b, \zeroterm, \oneterm), \oneterm)$.  Although this notation omits information about shared nodes (for example, from the notation above, we cannot know whether $1$-terminal nodes pointed to by $a$ and $b$ are shared or not), this does not affect the argument below.

We next define interpretation functions $\BDDI$ and $\ZDDI$ that transform a diagram into a combination set it represents as a BDD and a ZDD, respectively.

\begin{definition}
  \label{def:bdd-int}
  Given a set $X$, the interpretation $\beta_X\colon \D(X) \to \PP(X)$ of a diagram as a BDD is defined as follows:
  \begin{align*}
    \BDDI_X(\zeroterm) &= \emptyset, \qquad
                         \BDDI_X(\oneterm) = \P(X), \\
    \BDDI_X((a, F, G)) &= %(\neg a \land \BDDI(F)) \lor (a \land \BDDI(G))
                         \{ C \in \BDDI(F) \mid a \notin C \} \cup
                         \{ C \in \BDDI(G) \mid a \in C \}.
  \end{align*}
\end{definition}

The characteristic function of $\BDDI_X(F)$ is a function that receives $C \subseteq X$ and outputs either true or false, and the output can be computed by the following procedure.  Starting from the
root node of $F$, in each step the procedure looks up the label $x$ of the
current node, and if $x \notin C$ it explores the $0$-edge, and the $1$-edge otherwise.  The procedure
terminates when it reaches a terminal node, and returns true if its label is $1$, and false otherwise.

\begin{remark}
  A BDD is often introduced as a representation of a Boolean
  function, rather than a combination set.  In this case, the interpretation of the diagram in
  Figure~\ref{fig:ex_dd} is represented by a formula $a \lor b$ (this corresponds to a combination set $\{a, b, ab\}$; here, we used the convention introduced after Definition~\ref{def:comb}).  In general, the above definition can be rephrased as follows: $0$- and $1$-terminal nodes represent $\false$
  and $\true$, respectively, and a decision
  node $(a, F, G)$ represents $(\neg a \land \beta_X(F)) \lor (a
  \land \beta_X(G))$.
\end{remark}

\begin{definition}
  \label{def:zdd-int}
  Given a set $X$, the interpretation $\ZDDI_X\colon \D(X) \to \PP(X)$ of a diagram as a ZDD is defined as follows:
  \begin{align*}
    \ZDDI_X(\zeroterm) &= \emptyset, \qquad
                       \ZDDI_X(\oneterm) = \{\emptycomb\}, \\
    \ZDDI_X((a, F, G)) &= \ZDDI_X(F) \cup
                       \{ A \cup \{a\} \mid A \in \ZDDI_X(G) \}.
  \end{align*}
\end{definition}
This semantics is intuitively understood as follows.  When a diagram is seen as a ZDD, each path from the root to $\oneterm$ designates a single combination.  Such a path is called a \emph{$1$-path}.
On each decision node $(a, F, G)$, traversing its $1$-edge means ``include $a$,'' and $0$-edge means ``do not include $a$.''  Any element not mentioned in $p$ is not included in the combination.  Therefore, the combination a $1$-path $p$ represents is
\begin{equation*}
  \{a \in X \mid \text{$p$ contains a $1$-edge of a decision node labeled by $a$} \}.
\end{equation*}
Finally, $\ZDDI_X(F)$ is defined to be the set of all combinations that can be written in this form for some $1$-path $p$.  In particular, $\oneterm$ should be interpreted by the combination $\{\emptycomb\}$, because $\oneterm$ has only one, obvious $1$-path corresponding to $\emptycomb$.  This justifies the formal definition above.
It is straightforward to check that the interpretation of a diagram in Figure~\ref{fig:ex_dd} as a ZDD is $\{a, b\}$.  Unlike the case of BDD, the combination $ab$ is not included.

In the rest of the paper, when no confusion arises, we omit subscripts of $\BDDI$ and $\ZDDI$.

\subsection{Naturality of the Semantics}

What distinguishes $\BDDI$ and $\ZDDI$ is the treatment of elements that do not occur in a diagram.  This difference can be formally captured by Theorem~\ref{thm:dd-nat}, but we first give a more intuitive explanation.  This was essentially mentioned in the literature~\cite[\S~2--3]{Minato:1993:ZBS:157485.164890}, but let us summarize the main points in our terminology.

Let $F \in \D(X)$ and assume that there is an element $x \in X$ that do not occur in $F$.  When regarded as a BDD, this means that it does not matter whether a combination contains the element $x$ or not; more formally, $C \in \BDDI(F)$ if and only if $C \cup \{x\} \in \BDDI(F)$.  In contrast, when $F$ is regarded as a ZDD, $x$ not occurring in $F$ never appear in combinations in $\ZDDI(F)$, and therefore $C \in \ZDDI(F)$ only if $x \notin C$.
For example, let $F$ be the diagram in Figure~\ref{fig:ex_dd}.
% Then $F$ can be regarded as a diagram of both $X = \{a, b\}$ and $Y = \{a, b, c\}$, and we have
% \begin{equation*}
%   \BDDI_X(F) = \{a, b, ab\}, \enskip
%   \BDDI_Y(F) = \{a, b, ab, ac, bc, abc\}, \enskip
%   \ZDDI_X(F) = \{a, b\}, \enskip
%   \ZDDI_Y(F) = \{a, b\}.
% \end{equation*}
Then $F$ can be regarded as a diagram of $Y = \{a, b, c\}$, rather than $\{a, b\}$, and in this case we have
\begin{equation*}
  \BDDI_Y(F) = \{a, b, ab, ac, bc, abc\}, \text{ and }
  \ZDDI_Y(F) = \{a, b\}.
\end{equation*}
We can easily check that it is indeed the case that $C \in \BDDI_Y(F)$ if and only if $C \cup \{c\} \in \BDDI_Y(F)$, and $\ZDDI_Y(F)$ does not have any combination containing $c$.
% Then BDD-interpretation appears to change when $X$ is changed, whereas ZDD-interpretation does not.  However, if we write the right-hand side of each of these equations in a Boolean formula, then it is easily seen that the situation is different: the same formula represents both of the BDD-interpretations, but this is not true for ZDD-interpretations.

This behavior is nicely explained by the difference between the functor structures of $\PPc$ and $\PP$, and this is a formalized version of the assertion that ``a BDD represents a Boolean function, and a ZDD represents a combination set.''

Before stating the theorem, let us defined a functor structure of $\D$.  Given a map $f\colon X \to Y$, the action of $\D$ on $f$ is given by relabelling:
\begin{equation*}
  \D(f)(\zeroterm) = \zeroterm, \quad \D(f)(\oneterm) = \oneterm, \quad
  \D(f)((a, F, G)) = (f(a), \D(f)(F), \D(f)(G)).
\end{equation*}

Then, we can state our first main theorem.
\begin{theorem}
  \label{thm:dd-nat}
  $\BDDI$ is a natural transformation from $\D$ to $\PPc$, and $\ZDDI$ is a natural transformation from $\D$ to $\PP$.
\end{theorem}

\begin{remark}
  \label{rmk:dd-not-nat}
  $\BDDI$ is not a natural
  transformation from $\D$ to $\PP$, and $\ZDDI$ is not a natural
  transformation from $\D$ to $\PPc$.  For example, let $F$ be a diagram in Figure~\ref{fig:ex_dd}, and consider $X = \{a, b\}$, $Y = \{a, b, c\}$, and $i \colon X \hookrightarrow Y$.  Then we can easily check that $\BDDI_Y \circ \D(i) (F) \ne \PP(i) \circ \BDDI_X(F)$ and $\ZDDI_Y \circ \D(i) (F) \ne \PPc(i) \circ \ZDDI_X(F)$.
\end{remark}

% \begin{remark}
%   When $X$ is infinite, $\BDDI$ and $\ZDDI$ are not surjective.  Their images are the Boolean algebra and the idempotent semiring, respectively, freely generated by $X$.  They are concretely presented as: $\mathrm{Im}(\BDDI_X) = \{ \bigcup_{i=1}^n U(A_i, B_i) \mid A_i, B_i \in \Pf(X)\}$ where $U(A, B) = \{ C \subseteq X \mid A \subseteq C \subseteq X \setminus B\}$, and $\mathrm{Im}(\ZDDI_X) = \Pf^2(X)$ (with an idempotent addition $\cup$, and a multiplication $\sqcup$ defined by $P \sqcup Q = \{A \cup B \mid A \in P, B \in Q\}$).
% \end{remark}

\subsection{Ordered BDDs and ZDDs}
\label{subsec:odd}

We did not assume any structure on the set $X$, but in practice, a total order $\le$ on $X$ is often introduced, and the order of occurrence in a diagram is required to be the same as the order $\le$ on $X$ (the advantage of this restriction is that it allows an efficient implementation of operations on Boolean functions~\cite{Bryant1986}).

\begin{definition}
  Let $(X, \leq)$ be a totally ordered set.  $F \in \D(X)$ is said to \emph{respect} $\le$ if, for all $x, y \in X$, $x < y$ holds whenever $y$ occurs as a descendant of $x$ in $F$.
\end{definition}
For example, the diagram in Figure~\ref{fig:ex_dd} respects a total order $\le$ if and only if the order satisfies $a < b$.

Let $\D'(X, \leq)$ be the set of all diagrams over $X$ respecting $\leq$.
Then $\D'$ becomes a functor from the category of total order and
strictly monotone maps.  The functor structure of $\D'$ can be given by restricting that of $\D$: $\D'(f) = \D(f)$ for strictly monotone $f$.  To check that $\D'$ is well-defined, it suffices to see that, for any strictly monotone $f\colon (X, \le_X) \to (Y, \le_Y)$, a diagram respecting $\le_X$ is mapped by $\D(f)$ to a diagram respecting $\le_Y$.  Then $\BDDI$ and $\ZDDI$ restricts to natural transformations from $\D'$.  To be more precise, we can state these as follows.

\begin{lemma}
  \label{lem:odd-subfunctor}
  Let $\TO$ be the category of total order and strictly monotone maps.
  Then $\D'\colon \TO \to \Sets$ is a subfunctor of $\D \circ U$, where $U\colon \TO \to \Sets$ is the forgetful functor.
\end{lemma}

\begin{theorem}
  \label{thm:odd-nat}
  Let us define $\BDDI'_{(X, \leq)} = \BDDI_X$ and $\ZDDI'_{(X, \leq)} = \ZDDI_X$.  Then $\BDDI'\colon \D' \to \PPc \circ U$ and $\ZDDI'\colon \D' \to \PP \circ U$ are natural.
\end{theorem}
\begin{proof}
  Let $i\colon \D' \to \D \circ U$ be the inclusion.  Then we can write $\BDDI' = \BDDI U \circ i$ (where $\BDDI U$ denotes the whiskering by $U$).  Then, the naturality of $\BDDI'$ is an immediate consequence of that of $\BDDI\colon \D \to \PPc$, which follows from Theorem~\ref{thm:dd-nat}.  The same argument works for $\ZDDI' = \ZDDI U \circ i$, too.
\end{proof}

This proof is a special case of the following observation, which is used again in the next section.  Theorem~\ref{thm:odd-nat} is a direct consequence of Lemma~\ref{lem:odd-subfunctor} and this general fact (with $Q = \PPc$, $D = \D$, $\alpha = \BDDI$, $\CatC = \TO$, and $D' = \D'$ in the case of BDD, and $Q = \PP$ and $\alpha = \ZDDI$ for ZDD).
\begin{remark}
  \label{rmk:general-nonsense}
  Consider the following data.
  \begin{itemize}
  \item a functor $Q\colon \Sets \to \Sets$,
  \item a functor (of diagrams without restriction) $D \colon \Sets \to \Sets$,
  \item a natural transformation (giving an interpretation of diagrams) $\alpha \colon D \to Q$,
  \item a category (of some structures) $\CatC$,
  \item a (forgetful) functor $U \colon \CatC \to \Sets$, and
  \item a functor (of diagrams respecting the structure of $\CatC$-objects) $D'\colon \CatC \to \Sets$ such that $D' \subseteq D \circ U$.
  \end{itemize}
  Then the interpretation $\alpha$ induces a natural transformation $\alpha U \circ i \colon D' \to Q \circ U$, where $i$ is the inclusion from $D'$ into $D \circ U$.
\end{remark}

\section{Sentential Decision Diagrams}
\label{sec:sdd}

In this section, we show that the results obtained in the previous section can be extended to \emph{Sentential Decision Diagrams (SDDs)}~\cite{Darwiche:2011:SNC:2283516.2283536} and \emph{Zero-suppressed SDDs (ZSDDs)}~\cite{Nishino:2016:ZSD:3015812.3015969}, which generalize BDDs and ZDDs, respectively.

\subsection{Sentential Decision Diagrams without Constraints}

We first consider SDDs and ZSDDs that are not constrained in the sense that the occurrences of variables are not restricted in any way.
An SDD is a representation of a decomposition of Boolean function~\cite{PipatsrisawatDarwiche10} (satisfying a certain condition), so we begin with its definition.
\begin{definition}
  Let $X$ and $Y$ be disjoint sets of variables, and $f$ a Boolean function over $X \cup Y$.  An \emph{$(X, Y)$-decomposition} of $f$ is a representation of $f$ by a Boolean formula of the form $f(X, Y) = \bigvee_{i=1}^n (p_i(X) \land s_i(Y))$, where $p_i$ and $s_i$ depend only on variables in $X$ and $Y$, respectively.  Each $p_i$ is called a \emph{prime}, and each $s_i$ is called a \emph{sub} of the decomposition.
\end{definition}
In the original definition of SDDs, following the definition of a decomposition above, it is required that $p_i$ and $s_i$ have no common variable.  This restriction is omitted in our first definition of SDDs and ZSDDs, and considered later in Section~\ref{subsec:sdd-vtree}.

There is a well-known special case of decomposition, called Shannon decomposition, where $X$ consists of a single variable, $n=2$, $p_1 = X$, and $p_2 = \neg X$.  This decomposition corresponds to a BDD, and in this sense, SDD is a generalization of BDD (and similarly for ZDD and ZSDD).  The name \emph{sentential} decision diagram comes from the fact that an SDD allows primes to be a general formula (sentence), rather than only literals~\cite{Darwiche:2011:SNC:2283516.2283536}.

% \todo{Do we have to say this here?  Can't we postpone this to later section?}
% There are several, well-behaved class of the decompositions introduced above.  For example, if the primes $p_1, \dots, p_n$ satisfies (1) $p_i \land p_j = \false$ whenever $i \ne j$, and (2) $p_1 \lor \dots \lor p_n = \true$, then they are called a \emph{partition}.  In this case, the representation $f(X, Y) = \bigvee_{i=1}^n (p_i(X) \land s_i(Y))$ allows us to compute the value of $f$ as follows:  given a concrete input $x$ and $y$ for variables $X$ and $Y$, there is a unique $i$ such that $p_i(x) = 1$, and for this $i$ we have $f(x, y) = s_i(y)$.  The original definition of an SDD~\cite{Darwiche:2011:SNC:2283516.2283536} and a ZSDD~\cite{Nishino:2016:ZSD:3015812.3015969} assume primes to be a partition.  This condition is further discussed in Section~\ref{subsec:sdd-partition}.

We first define SDDs and their interpretation as combination sets.
\begin{definition}
  \label{def:sdd}
  Let $X$ be a set.  We define SDDs over $X$ and the interpretation $\SDDI_X$ inductively as follows.
  \begin{itemize}
  \item $\top$ and $\bot$ are SDDs, and interpreted by $\P(X)$ and $\emptyset$, respectively.
  \item $x$ and $\neg x$ are SDDs for each $x \in X$, and interpreted by $\{C \subseteq X \mid x \in C\}$ and $\{C \subseteq X \mid x \notin C\}$, respectively.
  \item If $p_i, s_i$ are SDDs for $1 \le i \le n$, then $\{(p_1, s_1),
    \dots, (p_n, s_n)\}$ is an SDD, and interpreted by $\bigcup_{i=1}^n
    (\SDDI_X(p_i) \cap \SDDI_X(s_i))$.
  \end{itemize}
  An SDD of either the first or the second form is called a \emph{terminal}, and the third a \emph{decomposition}.  We regard a decomposition as a set of pairs, and thus the order of pairs is irrelevant.
  % \todo{Should I show a diagram representation of an SDD (and a ZSDD)?}
\end{definition}

A zero-suppressed variant of SDDs, called ZSDDs, are defined as follows~\cite{Nishino:2016:ZSD:3015812.3015969}.
\begin{definition}
  Let $X$ be a set.  We define ZSDDs over $X$ and the interpretation $\ZSDDI_X$ inductively as follows.
  \begin{itemize}
  \item $\bot$ and $\emptycomb$ are ZSDDs, and interpreted by $\emptyset$ and $\{\emptycomb\}$, respectively.
  \item $x$ and $\pm x$ are ZSDDs for each $x \in X$, and interpreted by
    $\{\{x\}\}$ and $\{\emptycomb, \{x\}\}$, respectively.
  \item If $p_i, s_i$ are ZSDDs for $1 \le i \le n$, then $\{(p_1, s_1),
    \dots, (p_n, s_n)\}$ is a ZSDD, and interpreted by $\bigcup_{i=1}^n
    (\ZSDDI_X(p_i) \sqcup \ZSDDI_X(s_i))$, where $\sqcup$ is defined by
    \begin{equation*}
      P \sqcup Q = \{A \cup B \mid A \in P, B \in Q\}.
    \end{equation*}
  \end{itemize}
\end{definition}

We can prove an analogue of Theorem~\ref{thm:dd-nat}.
Let $\SDD(X)$ and $\ZSDD(X)$ be the set of SDDs and ZSDDs, respectively, over $X$.  Given a map $f\colon X \to Y$, we define $\SDD(f)\colon \SDD(X) \to \SDD(Y)$ by relabeling:
\begin{gather*}
  \SDD(f)(\top) = \top, \quad \SDD(f)(\bot) = \bot, \quad
  \SDD(f)(x) = f(x), \quad \SDD(f)(\neg x) = \neg f(x), \\
  \SDD(\{(p_i, s_i)\}_{i=1}^n) = \{(\SDD(f)(p_i), \SDD(f)(s_i))\}_{i=1}^n.
\end{gather*}
$\ZSDD(f)$ is defined similarly.  Then $\SDD$ and $\ZSDD$ are functors from $\Sets$ to $\Sets$, and we can prove the following.
\begin{theorem}
  \label{thm:sdd-nat}
  $\SDDI$ is a natural transformation from $\SDD$ to $\PPc$, and $\ZSDDI$ is a natural transformation from $\ZSDD$ to $\PP$.
\end{theorem}
We can prove this using the fact that $\PPc(f)$ preserves all Boolean operations, and $\PP(f)$ preserves both $\bigcup$ and $\sqcup$.  In addition, an analogue of Remark~\ref{rmk:dd-not-nat} holds for SDDs and ZSDDs.

\subsection{SDDs and ZSDDs Respecting Vtrees}
\label{subsec:sdd-vtree}

In this section, we extend the result of Section~\ref{subsec:odd} to SDDs and ZSDDs.  A total order is replaced by a vtree defined below.
\begin{definition}
  \label{def:vtree}
  A \emph{vtree} for a set $X$ is a rooted, full binary tree whose leaves are in one-to-one correspondence with elements of $X$.
\end{definition}
Below, a leaf that corresponds to $x \in X$ is denoted simply by $x$, and a vtree with left and right children $v$ and $w$ is denoted by $(v, w)$.  We write $|v|$ for $X$ if $v$ is a vtree for $X$.

% We define SDDs and ZSDDs respecting a vtree.
% \begin{definition}
%   We define the set $\SDD'(v)$ of SDDs respecting a vtree $v$ inductively as follows.
%   \begin{enumerate}
%   \item $\top, \bot \in \SDD'(v)$ for any $v$.
%     \label{item:const_term}
%   \item $x, \neg x \in \SDD'(v)$ for any $v$ such that $x \in |v|$.
%     \label{item:var_term}
%   \item If $\alpha \in \SDD'(v) \cup \SDD'(w)$, then $\alpha \in \SDD'((v, w))$.
%     \label{item:congruence}
%   \item If $p_i \in \SDD'(v)$ and $s_i \in \SDD'(w)$ for $1 \le i \le n$, then $\{(p_1, s_1), \dots, (p_n, s_n)\} \in \SDD'((v, w))$.
%     \label{item:decomposition}
%   \end{enumerate}
% \end{definition}
\begin{definition}
  We define an SDD or a ZSDD \emph{respecting} $v$ inductively as follows.
  \begin{enumerate}
  \item $\top$, $\bot$, $\emptycomb$ respect any vtree.
    \label{item:const_term}
  \item $x$, $\neg x$, $\pm x$ respect a leaf corresponding to $x$.
    \label{item:var_term}
  \item If $\alpha$ respects either $v$ or $w$, then $\alpha$ respects $(v, w)$.
    \label{item:congruence}
  \item If $p_i$ respects $v$ and $s_i$ respects $w$ for each $1 \le i \le n$, then $\{(p_1, s_1), \dots, (p_n, s_n)\}$ respects $(v, w)$.
    \label{item:decomposition}
  \end{enumerate}
\end{definition}
This definition is equivalent to that by Bova~\cite{Bova:2016:SEM:3015812.3015951} (except that we do not require primes to form a partition), and relaxes the ones by
Darwiche~\cite[Def.\ 5]{Darwiche:2011:SNC:2283516.2283536} and by Nishino et al.~\cite[Def.\ 1]{Nishino:2016:ZSD:3015812.3015969} (they do not allow rule~\ref{item:congruence}).
% As far as the authors are aware of, two definitions do not make significant difference.  In particular, both of the definitions allow us to prove uniqueness of the compressed and trimmed representation.  However, it appears that the existence requires relaxed version.

We next define a category having vtrees as objects, and embeddings defined below as morphisms.
\begin{definition}
  Let $v$ and $w$ be vtrees.  We define \emph{embeddings (of vtrees)} from $v$ to $w$ inductively.  Let $f\colon |v| \to |w|$ be a map.
  \begin{itemize}
  \item $f$ is an embedding if $v$ is a leaf.
  \item $f$ is an embedding if $w = (w_1, w_2)$, and $f$ is an embedding from $v$ to either $w_1$ or $w_2$.
  \item $f$ is an embedding if $v = (v_1, v_2)$, $w = (w_1, w_2)$, and the restriction of $f$ to $|v_i|$ is an embedding from $v_i$ to $w_i$ for $i=1,2$.
  \end{itemize}
\end{definition}
It is not difficult to check that embeddings are closed under composition.  We write $\VTree$ for the category of vtrees and embeddings.

\begin{remark}
  A vtree generalizes a finite total order: $x_1 < x_2 < \dots < x_n$ corresponds to a vtree $(x_1, (x_2, \dots (x_{n-1}, x_n)\dots))$.  If $X$ and $Y$ are finite totally ordered sets, and $v_X$ and $v_Y$ are the corresponding vtrees, then $f\colon X \to Y$ is strictly monotone if and only if it is an embedding from $v_X$ to $v_Y$.  Moreover, every BDD respecting the order of $X$ can be translated to an SDD respecting $v_X$~\cite{Darwiche:2011:SNC:2283516.2283536}, and similarly for ZDDs and ZSDDs~\cite{Nishino:2016:ZSD:3015812.3015969}.
\end{remark}

Let $\SDD'(v)$ be the set of SDDs respecting $v$.  Then, for an embedding $f\colon v \to w$, it is straightforward to check that $\SDD(f)$ restricts to a map $\SDD'(f)\colon \SDD'(v) \to \SDD'(w)$, that is, $\SDD'$ is a functor from $\VTree$ to $\Sets$.  In the same manner, we can define a functor $\ZSDD'$ such that $\ZSDD'(v)$ is the set of ZSDDs respecting $v$.  The following is an analogue of Lemma~\ref{lem:odd-subfunctor}.

\begin{lemma}
  \label{lem:vt-subf}
  Let $U \colon \VTree \to \Sets$ be the forgetful functor, which maps $v$ to $|v|$.  Then, $\SDD'$ and $\ZSDD'$ are subfunctors of $\SDD \circ U$ and $\ZSDD \circ U$, respectively.
\end{lemma}

Then, similarly to Theorem~\ref{thm:odd-nat}, the following is an immediate corollary of Remark~\ref{rmk:general-nonsense}.

\begin{corollary}
  \label{cor:vt-nat}
  $\SDDI$ and $\ZSDDI$ restricts to natural transformations from $\SDD'$ and $\ZSDD'$, respectively.  Concretely, the restrictions are given by $\SDDI'_v = \SDDI_{|v|}$ and $\ZSDDI'_v = \ZSDDI_{|v|}$ for a vtree $v$, and are natural transformations from $\SDD'$ to $\PPc \circ U$ and from $\ZSDD'$ to $\PP \circ U$, respectively.
\end{corollary}

\subsection{Strong Determinism and Partition}
\label{subsec:sdd-partition}

The original definitions of SDDs and ZSDDs require primes to form a partition~\cite{Darwiche:2011:SNC:2283516.2283536,Nishino:2016:ZSD:3015812.3015969}.  Such a restriction can be treated in a similar manner to the previous section.

\begin{definition}
  Let $\alpha \in \SDD(X)$ be an SDD\@.
  \begin{enumerate}
  \item $\alpha$ is \emph{strongly deterministic} if every decomposition of $\alpha$ has pairwise disjoint primes.  More concretely, every decomposition $\{(p_i, s_i)\}_{i=1}^n$ in $\alpha$ satisfies $\SDDI(p_i) \cap \SDDI(p_j) = \emptyset$ whenever $i \ne j$.  (A decomposition of a Boolean function having this property is said to be strongly deterministic~\cite{PipatsrisawatDarwiche10}.)
  \item $\alpha$ is a \emph{partition SDD} if it is strongly deterministic and every decomposition $\{(p_i, s_i)\}_{i=1}^n$ in $\alpha$ satisfies $\SDDI(p_1) \cup \dots \cup \SDDI(p_n) = \P(X)$.
  \end{enumerate}
\end{definition}

\begin{definition}
  A ZSDD $\alpha \in \ZSDD(X)$ is \emph{strongly deterministic} if every decomposition of $\alpha$ has pairwise disjoint primes.
\end{definition}

Similarly to Sections~\ref{subsec:odd} and \ref{subsec:sdd-vtree}, we have the following.
\begin{lemma}
  Let $f\colon X \to Y$ be a map.
  \begin{enumerate}
  \item $\SDD(f)$ preserves strongly deterministic SDDs and partition SDDs.
  \item If $f$ is injective, then $\ZSDD(f)$ preserves strongly deterministic ZSDDs.
  \end{enumerate}
\end{lemma}
\begin{proof}
  By induction.
  \begin{enumerate}
  \item Let $\alpha = \{(p_i, s_i)\}_{i=1}^n \in \SDD(X)$.  We first show that if $\{p_i\}_i$ are pairwise disjoint, then so are $\{\SDD(f)(p_i)\}_i$.  Let us assume $\SDDI(p_i) \cap \SDDI(p_j) = \emptyset$.  Then by naturality of $\SDDI$ we have
    \begin{align*}
      \SDDI(\SDD(f)(p_i)) \cap \SDDI(\SDD(f)(p_j))
      & = \PPc(f)(\SDDI(p_i)) \cap \PPc(f)(\SDDI(p_j)) \\
      & = \PPc(f)(\SDDI(p_i) \cap \SDDI(p_j)) \\
      & = \emptyset.
    \end{align*}
    Similarly we can easily check that $\bigcup_i \SDDI(p_i) = \P(X)$ implies $\bigcup_i \SDDI(\SDD(f)(p_i)) = \P(Y)$.
  \item Let $\alpha = \{(p_i, s_i)\}_{i=1}^n \in \ZSDD(X)$.  We show that $\{p_i\}_i$ are pairwise disjoint, then so are $\{\ZSDD(f)(p_i)\}_i$.  This can be done in the same manner as above, except that $\SDDI$ and $\PPc$ are replaced by $\ZSDDI$ and $\PP$, respectively.  Notice that $\PP(f)$ preserves the intersection, which follows from the assumption that $f$ is injective.
    \qedhere
  \end{enumerate}
\end{proof}
This means that there are subfunctors of $\SDD$ taking only partition SDDs and strongly deterministic SDDs.  Similarly for strongly deterministic ZSDDs, but the domain of the functor should be the subcategory of $\Sets$ whose morphisms are injections.  In the same manner as Theorem~\ref{thm:odd-nat} and Corollary~\ref{cor:vt-nat}, $\SDDI$ and $\ZSDDI$ restrict to these subfunctors.

We could also define a partition ZSDD, but it appears that such a notion is, unlike strong determinism, not preserved by $\ZSDD(f)$.  By adapting the original definition of ZSDDs~\cite{Nishino:2016:ZSD:3015812.3015969} to the current context, we could define: a ZSDD $\alpha$ is a partition ZSDD with respect to a vtree $v$ if, for any of its decomposition $\beta$, there is a subtree $(v_1, v_2)$ of $v$ respected by $\beta$ and the primes of $\beta$ form a partition of $|v_1|$.  However, such a notion would not be preserved by a vtree embedding.  Indeed, consider $\alpha = \{(a, \emptycomb), (\emptycomb, b)\} \in \ZSDD(\{a, b\})$.  The primes $a$ and $\emptycomb$ of $\alpha$ denote $\{\{a\}\}$ and $\{\emptycomb\}$, whose union is $\P(\{a\})$.  Therefore this $\alpha$ is a partition ZSDD with respect to a vtree $(a, b)$.  However, if we consider a vtree $((a, c), b)$, into which $(a, b)$ can be embedded, primes of $\alpha$ do not form a partition.  For this reason, we do not further consider partition ZSDDs in this paper.

From this observation, we can conclude that both partition and strong determinism of Boolean functions (or SDDs) are well-behaved, but only strong determinism is so for combination sets (or ZSDDs).  Nishino et al.\ introduced a notion of implicit partition~\cite[Def.~7]{Nishino:2016:ZSD:3015812.3015969}, which is roughly the same as strong determinism.  The discussion above suggests that implicit partition would be a more sensible notion than partition, which was used in the first definition of ZSDDs.

\section{Conclusion}
\label{sec:concl}

We investigated the difference between BDDs and ZDDs, as well as their variants, and identified a formally stated theorem that captures the fact that a BDD represents a Boolean function and a ZDD represents a combination set.  This is done by observing that their
definitions reflect the actions on morphisms of two functors $\PPc$
and $\PP$, respectively.  In addition, we have observed that similar
result holds for SDD and ZSDD as well.
% This suggests that it is a common
% phenomenon that there are two variants of data
% structures, reflecting the two functor structures corresponding to Boolean functions and combination sets.

There are many other types of decision diagrams in the literature.  For example, a sequence BDD~\cite{Loekito2010} and $\pi$DD~\cite{Minato11:pidd} are proposed as representations of a set of sequences and a set of permutations, respectively.  It would be interesting to consider whether these data structures have similar naturality property with respect to appropriate functors.  It is also an interesting future work to investigate whether there exists a general principle to design a natural representation, when a class of data is specified as a functor.
% We should also investigate operations on representations such as intersection and union, and what can be said about these operations from a categorical point of view.

\section*{Acknowledgments}

This work was supported by JST CREST Grant Number JPMJCR1401, Japan.

\ifelsarticle
\bibliographystyle{elsarticle-num}
\else
\bibliographystyle{plain}
\fi
\bibliography{dd}

\begin{thebibliography}{10}

\bibitem{Akers78}
Sheldon~B. Akers.
\newblock Binary decision diagrams.
\newblock {\em {IEEE} Transactions on Computers}, C-27(6):509--516, June 1978.

\bibitem{Bova:2016:SEM:3015812.3015951}
Simone Bova.
\newblock {SDD}s are exponentially more succinct than {OBDD}s.
\newblock In {\em Proceedings of the Thirtieth AAAI Conference on Artificial
  Intelligence}, AAAI'16, pages 929--935. AAAI Press, 2016.

\bibitem{Bryant1986}
Randal~E. Bryant.
\newblock Graph-based algorithms for {B}oolean function manipulation.
\newblock {\em IEEE Trans. Comput.}, 35(8):677--691, August 1986.

\bibitem{Darwiche:2011:SNC:2283516.2283536}
Adnan Darwiche.
\newblock {SDD}: A new canonical representation of propositional knowledge
  bases.
\newblock In {\em Proceedings of the Twenty-Second International Joint
  Conference on Artificial Intelligence - Volume Volume Two}, IJCAI'11, pages
  819--826. AAAI Press, 2011.

\bibitem{Lee59}
C.~Y. Lee.
\newblock Representation of switching circuits by binary-decision programs.
\newblock {\em The Bell System Technical journal}, 38(4):985--999, July 1959.

\bibitem{Loekito2010}
Elsa Loekito, James Bailey, and Jian Pei.
\newblock A binary decision diagram based approach for mining frequent
  subsequences.
\newblock {\em Knowledge and Information Systems}, 24(2):235--268, Aug 2010.

\bibitem{Minato:1993:ZBS:157485.164890}
Shin-ichi Minato.
\newblock Zero-suppressed {BDD}s for set manipulation in combinatorial
  problems.
\newblock In {\em Proceedings of the 30th International Design Automation
  Conference}, DAC '93, pages 272--277, New York, NY, USA, 1993. ACM.

\bibitem{Minato11:pidd}
Shin-ichi Minato.
\newblock $\pi${DD}: A new decision diagram for efficient problem solving in
  permutation space.
\newblock In Karem~A. Sakallah and Laurent Simon, editors, {\em Theory and
  Applications of Satisfiability Testing - SAT 2011}, pages 90--104, Berlin,
  Heidelberg, 2011. Springer Berlin Heidelberg.

\bibitem{Nishino:2016:ZSD:3015812.3015969}
Masaaki Nishino, Norihito Yasuda, Shin-ichi Minato, and Masaaki Nagata.
\newblock Zero-suppressed sentential decision diagrams.
\newblock In {\em Proceedings of the Thirtieth AAAI Conference on Artificial
  Intelligence}, AAAI'16, pages 1058--1066. AAAI Press, 2016.

\bibitem{PipatsrisawatDarwiche10}
Knot Pipatsrisawat and Adnan Darwiche.
\newblock A lower bound on the size of decomposable negation normal form.
\newblock In {\em Proceedings of the Twenty-Fourth AAAI Conference on
  Artificial Intelligence}, AAAI'10, pages 345--350. AAAI Press, 2010.

\end{thebibliography}

\end{document}